\documentclass[preprint,12pt,sort&compress]{elsarticle}
\usepackage[left=2.0cm,right=2.0cm,top=1cm,bottom=2cm,bindingoffset=0cm]{geometry}
\usepackage{eepic}
\usepackage{amsmath,amssymb,amsthm}
\usepackage{multicol}
\usepackage{subfig}
\usepackage{enumitem}

\newtheorem{theorem}{Theorem}

\newtheorem{proposition}{Proposition}

\usepackage[T2A]{fontenc}
\usepackage[english]{babel}

\usepackage{listings, xcolor}

\DeclareMathOperator{\LEs}{LE}

\begin{document}

\begin{frontmatter}

 \title{Study of irregular dynamics in an economic model: \\ attractor localization and {L}yapunov exponents}

\author[hse]{Tatyana~A.~Alexeeva\,}
 \author[spbu,fin,ipm]{Nikolay~V.~Kuznetsov\corref{cor}}
 \ead{Corresponding author: nkuznetsov239@gmail.com}
 \author[spbu]{Timur~N.~Mokaev\,}

 \address[hse]{St. Petersburg School of Mathematics, Physics  and Computer Science, HSE University, \\
 194100 St. Petersburg, Kantemirovskaya ul., 3, Russia}
 \address[spbu]{Faculty of Mathematics and Mechanics,
 St. Petersburg State University, 198504 Peterhof,
 St. Petersburg, Russia}
 \address[fin]{Department of Mathematical Information Technology,
 University of Jyv\"{a}skyl\"{a},  40014 Jyv\"{a}skyl\"{a}, Finland}
 \address[ipm]{Institute for Problems in Mechanical Engineering RAS, 199178 St. Petersburg, V.O., Bolshoj pr., 61, Russia}

\begin{abstract}

Cyclicity and instability inherent in the economy can manifest themselves in irregular fluctuations, including chaotic ones, which significantly reduces the accuracy of forecasting the dynamics of the economic system in the long run.
We focus on an approach, associated with the identification of a deterministic endogenous mechanism of irregular fluctuations in the economy. 
Using of a mid-size firm model as an example, we demonstrate the use of effective analytical and numerical procedures for calculating the quantitative characteristics of its irregular limiting dynamics based on Lyapunov exponents, such as dimension and entropy. 
We use an analytical approach for localization of a global attractor and study limiting dynamics of the model.
We estimate the Lyapunov exponents and get the exact formula for the Lyapunov dimension of the global attractor of this model analytically. 
With the help of delayed feedback control (DFC), the possibility of transition from irregular limiting dynamics to regular periodic dynamics is shown to solve the problem of reliable forecasting.
At the same time, we demonstrate the complexity and ambiguity of applying numerical procedures to calculate the Lyapunov dimension along different trajectories of the global attractor, including unstable periodic orbits (UPOs).

\end{abstract}

\begin{keyword}
	Lyapunov exponents, Lyapunov dimension, chaos, unstable periodic orbit, absorbing set, mid-size firm model

 \end{keyword}

\end{frontmatter}

\section{Introduction}

Increasing uncertainty, unpredictability, and instability in the world, nature cataclysms, a series of economic crises, self-fulfilling expectations which give rise to bubbles and crashes, as well as rapid development and implementation of digital technologies in everyday life have posed a number of new challenges for scientists, governments, and policy makers: to study, understand and interpret
the behavior of complex dynamical systems, including socio-economic model \cite{Scheffer_etal21-2012,Battiston_et.al.22-2016,KindlebergerA-book-2015}.

An inherent component of observed economic processes is cyclicality, which is manifested through the occurrence of various types of fluctuations in the economic system under consideration.
In particular, regular fluctuations could be either periodic boom-bust phenomena associated with predictable changes in some elements of the economic system that reappear at fairly constant time intervals, or seasonal fluctuations that are permanent in nature. Regular and stable periodic oscillations lead to the predictable dynamics of the process’ model and are quite simple to describe mathematically. A number of straightforward quantitative measures, such as phase-frequency characteristics and amplitude, can be calculated for them.

However more offen, economic systems exhibit irregular (including chaotic) behavior. 
The role of irregular oscillatory dynamics for forecasting and stabilization of economic processes significantly depends on the source and nature of these fluctuations. On the one hand, irregular economic fluctuations could be the result of unusual events such as large bankruptcies, oil and currency shocks, floods, strikes, civil unrest, epidemics, etc. These events could be thought of as initiated by exogenous shocks. On the other hand, irregular fluctuations could be generated by endogenous mechanisms inherent in the very nature of economic systems.
Thus, there are two ways to examine of irregularity in the economy. First approach takes into account random processes that are considered in the model as exogenous shocks.  Second one is based on identification of a deterministic endogenous mechanism of occurrence of irregular fluctuations, which may also be chaotic.
These two approaches were developed in economics literature in parallel and generated a lot of discussion regarding the views on the sources of irregular fluctuations (see, e.g. \cite{Benhabib-2008-2016,AkhmetAF-JEBO-2014,BeaudryGP-AER-2020}). 

Since the 1970s, there has been keen interest in the study of deterministic chaotic dynamics in economic models within the framework of the second approach. This research was stimulated by the discovery of chaos in dynamical systems by Lorenz and Ueda \cite{Lorenz-1963,UedaAH-1973}.
Many famous economists (see, e.g. \cite{BenhabibN7-1979,BenhabibN-JET-1985,Day-AER-1982,Day2-1983,Grandmont-Ec-1985,BoldrinLM1-1986,BoldrinW-1990,DayS-JEBO-1987,BaumolB-JEP-1989,Medio6-1992,Sorger-JET-1992,Mitra-JET-1996,BrockHomEc-1997,BrockH9-1998,MitraS-1999,Rosser-Book-2000,BarnettSert10-2000,BenhabibSGU-AER-2002,Slobodyan-JM-2007,Benhabib-2008-2016}) have suggested numerous examples of economic models in which qualitatively and quantitatively reasonable irregular fluctuations might occur in purely deterministic settings.
For instance, the larger literature \cite{BenhabibN7-1979,BenhabibD-EL-1980,BenhabibD-JEDC-1982,Grandmont-Ec-1985,BoldrinLM1-1986,DeneckereP-JET-1986,BenhabibR-JET-1990,Sorger-JET-1992,Mitra-JET-1996} examines the endogenous cycles and irregular chaotic dynamics which could be generated by deterministic, equilibrium models of the economy.
The models often exhibit complex dynamics characterized by both chaotic behavior and instability. Such combination suggests a nonlinear dynamical system, somewhat unstable at the core, but effectively contained further out. 
The contribution of these models has been to demonstrate the compatibility of endogenous irregular fluctuations with equilibrium dynamics in economics.
At the same time, theoretical tools were developed for effective chaos control, which, by small fine-tuning the parameters of system, made it possible to stabilize selected orbits embedded in a chaotic attractor and nudge the dynamics toward a desired trajectory. Examples applications of these tools can be found in \cite{HolystHHW-1996,HolystU-2000,Kopel-1997,Kaas-1998,BalaMM-1998,BoccalettiGLMM-2000,MendesM-2005,ChenChen-2007,SalariehAl:3-2009}.
The reviewed literature shows the relevance of chaos for economic models and contributed to development of advanced mathematical tools for study of complex nonlinear dynamical systems in economics, which continues up to now. During the last few years, highly influential authors published a number of significant papers (see, e.g. \cite{SlobodyanW-AEJM-2012,Bella-CSF-2017,BellaMV-JET-2017,SandersFarmerG-SR-2018,Maliar-CERP-2018,BoccalettiLi-NJPh-2020,BellaLMV-CSF-2020,Barnett-2020,AchdouHLLM-RES-2021}). 
The studies of models with irregular dynamics have received a new impetus and spread into many subfields of economic theory. Especially, such models offer important contributions in macroeconomics, dynamical game theory, theory of rational inattention, finance, environmental economics, and industrial organization (for survey of the literature, see \cite{AnufrievRT-2018}).

To understand, describe and make measurable the properties of irregular dynamics it is important to calculate its quantitative characteristics. Indicators based on Lyapunov exponents, including such as entropy and dimension, naturally arise in economics \cite{DictionaryEc-2018}. In economic models these characteristics could be considered as indicators of irregular (primarily, chaotic) behavior, as the growth rate of the value of some economic variable (for instance, technology level), or as a measure of costs of making decisions by a rationally inattentive agent who acquires information about the values of alternatives through a limited-capacity channel (see, e.g. \cite{DechertG-JAE-1992,MontrucchioS-JMEc-1996,HommesBook12-2006,MatejkaMW-JET-2018}).
In this paradigm important results and arguments were presented which provide novel support for the idea that business cycles may be largely driven by endogenous deterministic cyclical forces (see,~e.g. \cite{BeaudryGP-AER-2020,AsanoKHFarmer-PNAS-2021,Galizia-QE-2021}).

There are two main approaches in studying this topic.
The first approach is based on the possibility of obtaining analytical results for low-dimensional nonlinear models (in the literature, two-dimensional dynamical systems are most often studied). The second one is based on the ability to study complex irregular dynamics using numerical procedures.
However, the possibility of obtaining reliable results using them is significantly limited due to the necessity of performing calculations only over finite time intervals, rounding-off errors in numerical methods, and the unbounded space of initial data sets \cite{LeonovKM-2015-EPJST,KuznetsovLMPS-2018,KuznetsovKKMD-2019-IFAC,KuznetsovMKK-2020,AlexeevaBKM-CSF-2020,AlexeevaBKM-IFAC-2020}. 
It should be noted that the sensitivity to small changes in the initial data, inherent in irregular (chaotic) dynamics, can cause significant forecasting errors. This, on the one hand, can explain some of the difficulties associated with forecasting behavior of the models, and on the other hand could be interpreted as unpredictability in real world problems (see, e.g. \cite{BeaudryGP-AER-2020}). Trajectories in models of such processes may be attracted not to a stationary point or a periodic cycle, but to an irregular invariant set, including chaotic attractor.
Additional complexity of the dynamics can be also associated with various unstable orbits embedded into the chaotic attractor of the dynamical system. Stabilization of unstable orbits makes it possible to improve the forecasting of the model dynamics \cite{AlexeevaBKM-IFAC-2020}.
Analytical methods allow overcoming these limitations at least for some low-dimensional models (see, e.g. \cite{LeonovKKK-2016-CNSCS,AlexeevaBKM-CSF-2020}) and are able to mitigate the influence of computer errors. Thus, this is capable of making reliable forecasts of model dynamics and of getting its exact qualitative and quantitative characteristics.

We continue the line of research on the limiting dynamics for mid-size firm model, which began in \cite{AlexeevaBKM-CSF-2020,AlexeevaBKM-IFAC-2020}, where we have obtained conditions for the global stability. 
In this paper we focus on a different approach, associated with the identification of deterministic endogenous mechanisms of irregular fluctuations in economic systems.
We use an analytical approach for localization of a global attractor and study limiting dynamics of the model.
We estimate the Lyapunov exponents and get the exact formula for the Lyapunov dimension of the global attractor of this model analytically. 
With the help of DFC, the possibility of transition from irregular limiting dynamics to regular periodic dynamics is shown to solve the problem of reliable forecasting.
At the same time, we demonstrate the complexity and ambiguity of applying numerical procedures to calculate the Lyapunov dimension along different trajectories of the global attractor, including UPOs.

\section{Problem statement}

For understanding and reliable predicting the behavior of 
 economic models in continuous time the study of its limit oscillations is an important task. 
This task could be solved by an analytical localization of the global attractor (whenever applicable) for the corresponding system of ODE, i.e., constructing a bounded closed positively invariant region (an absorbing set). 
On this attractor, along with the corresponding solution for the system
we obtain some estimates of irregular (including chaotic) dynamics. This allows us to calculate various quantitative characteristics based on the Lyapunov exponents such as the Lyapunov dimension of the attractor and entropy.

Consider the Sharovalov model proposed in \cite{ShapovalovKBA-2004} which describes the behavior of a mid-size firm

\begin{equation}
\begin{cases}
\begin{aligned}
&\dot x=-\sigma x+\delta y,\\
&\dot y=\mu x +\mu y-\beta xz,\\
&\dot z=-\gamma z+\alpha xy,
\end{aligned}
\label{SysShap6}
\end{cases}
\end{equation}
where coefficients $\alpha$, \, $\beta$, \, $\sigma$, \, $\delta$, \, $\mu$, \, $\gamma$ at variables $(x,\,y,\,z)$ are positive control parameters with the economic meaning. We define this model in terms of the differences between actual levels of the variables $X$, $Y$, and $Z$, denoted the growth of three main factors of production: the loan amount $X$, fixed capital $Y$ and the number of employees $Z$ (as an increase in human capital), and its potential (natural) levels $x_p$, $y_p$, and $z_p$ respectively\footnote{We assume that the potential (natural) levels of factors of production correspond to the production possibilities of a mid-size firm as a whole, reflecting its natural, technological, and institutional constraints.}\!. 
Thus, we consider the gap between the actual and potential levels of factors of production: $x = X - x_p$, $y= Y- y_p$, and $z= Z - z_p$, where $X$, $Y$, and $Z$ are nonnegative. 
Note that system \eqref{SysShap6} describes the behavior of a mid-size firm correctly when the global attractor or its absorbing set lays in the domain $x\ge -{x}_{p}$, $y\ge -{y}_{p}$, and $z\ge -{z}_{p}$.

System \eqref{SysShap6} can be reduced to a Lorenz-like system

\begin{equation}
\begin{cases}
\begin{aligned}
&\dot x= - c x + c y, \\
&\dot  y=r x + y - x z, \qquad \mbox{where} \,\, c = \frac{\sigma}{\mu} , r=\frac{\delta}{\sigma}, b = \frac{\gamma}{\mu},\\
&\dot z=-b z + xy, \end{aligned}
\label{SysLorenz3}
\end{cases}
\end{equation}
using the following coordinate transformation

\begin{equation}
(x, y, z) \rightarrow  \left(\frac{\mu}{\sqrt{\alpha \beta}} x,\, \frac{\mu \sigma}{\delta \sqrt{\alpha \beta}} y, \, \frac{\mu \sigma}{\delta \beta} z\right), \, t \rightarrow \frac{t}{\mu}.
\label{TrfShLor}
\end{equation}
System \eqref{SysLorenz3} in crucial respect differs from the classical Lorenz system \cite{Lorenz-1963} in the sign of the coefficient at $y$ in the second equation, which is 1 here and -1 in the Lorenz system.

Accordingly, the inverse transformation

\begin{equation}
(x, y, z) \rightarrow  \left(\frac{\sqrt{\alpha \beta}}{\mu} x, \frac{r \sqrt{\alpha \beta}}{\mu} y, \frac{r \beta}{\mu} z\right),\, t \rightarrow \mu t
\label{InvtrfLorSh}
\end{equation}
reduces system \eqref{SysLorenz3} to system \eqref{SysShap6} with coefficients $\sigma = c \mu, \delta = r c \mu, \gamma = b \mu$ \footnote{Transformations \eqref{TrfShLor} and \eqref{InvtrfLorSh} do not change the direction of time, which is essential for analysis of the Lyapunov dimension and Lyapunov exponents \cite{LeonovK-2015-AMC}.}.

In addition, system \eqref{SysShap6} with parameters satisfying the relations $\sigma^2 / (\sigma- \delta) = \mu $ and $ \delta <\sigma <\mu$ can be reduced to the well-known Chen system \cite{ChenU-1999}

\begin{equation}
\begin{cases}
\begin{aligned}
&\dot x= - d x + d y, \\
&\dot  y= (c - d) x + c y - x z, \qquad \mbox{with} \,\, d = \sigma , c=\frac{\sigma^2}{\sigma - \delta} = \mu, b = \gamma, \, d<c,\\
&\dot z= - b z + x y, \end{aligned}
\label{SysChen}
\end{cases}
\end{equation}
using coordinate substitutions

\begin{equation}
(x, y, z) \rightarrow  \left(\frac{1}{\sqrt{\alpha \beta}} x,\, \frac{\sigma}{\delta \sqrt{\alpha \beta}} y, \, \frac{\sigma}{\delta \beta} z\right).
\label{Trf:ShapChen}
\end{equation}

The possibility of reducing system \eqref{SysShap6}  to the Chen system \eqref{SysChen} under the above conditions shows the complexity of studying a mid-size firm model. The problem of analytical calculation of the dimension of the attractor for the Chen system remains an issue \cite{LeonovK-2015-AMC}.

It was shown in \cite{AlexeevaBKM-CSF-2020} that for system \eqref{SysLorenz3} the global absorbing set  $\mathcal{B}= \Omega_1 \bigcap \mathcal{B}_R$ can be constructed under conditions $2 < b < 2c$ (Fig.~\ref{fig:shap:abs-set}), where $\Omega_1 = \left\{(x,y,z) \in \mathbb{R}^3 ~|~  z \geq \frac{x^2}{2 c}\right\}$ is the parabolic cylinder, $\mathcal{B}_R =
 \left\{\!(x,y,z)\,\in \mathbb{R}^3 ~|~
  \frac{1}{2} \left[ B^2 x^2 - 2 \, B \,x \, y + y^2
    + \big(z - \big(r + (B^2 + B) c - B\big)\big)^2 \right] \leq \eta \! \right\}$ is the ellipsoid, $B=\tfrac{1}{2}\big(\tfrac{1}{c} + \tfrac{b}{2c}\big)$, and $\eta$ is a chosen parameter.

 \begin{figure}[ht]
  \centering
  \includegraphics[width=0.5\linewidth]{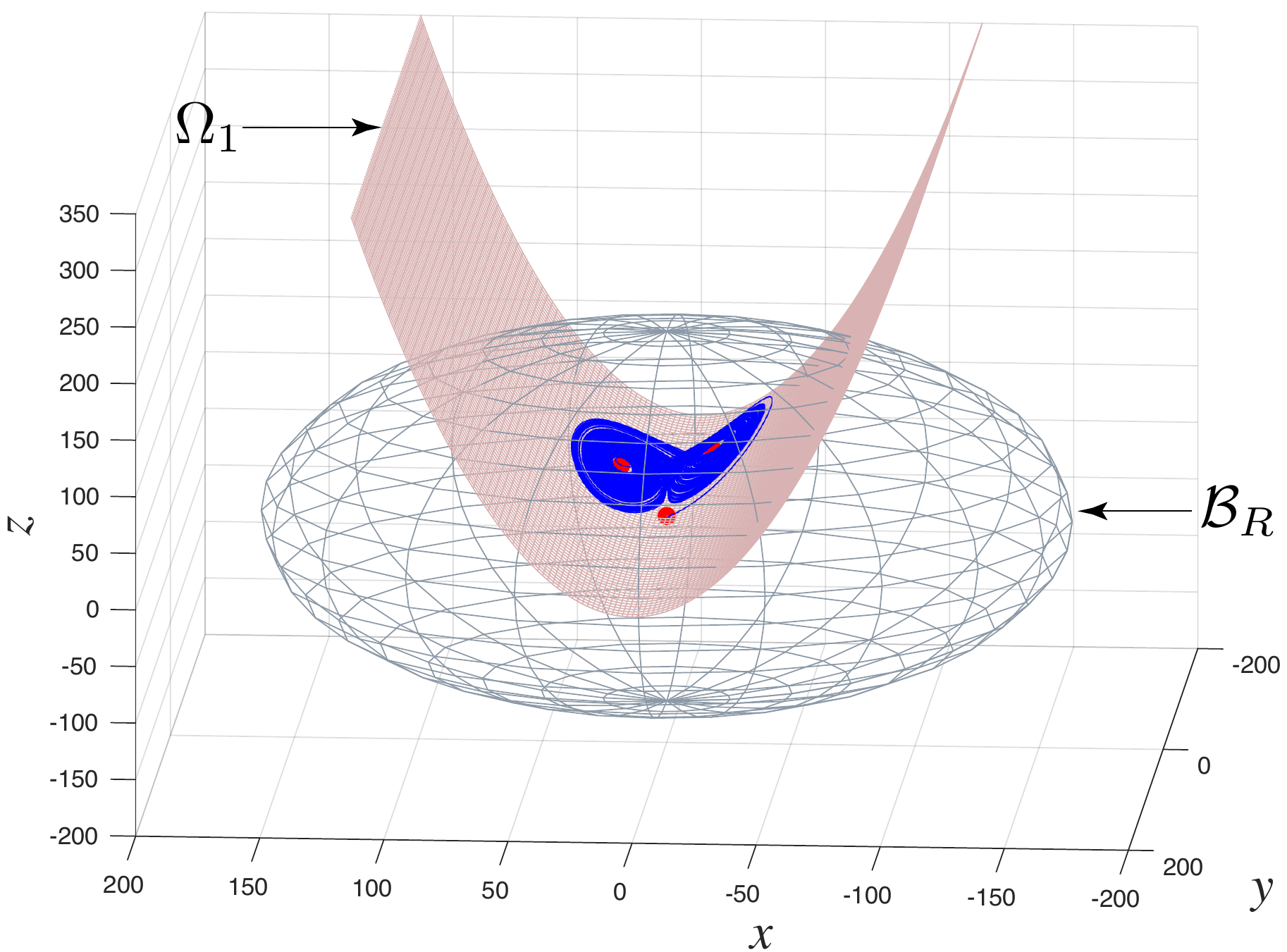}
  \caption{Analytical localization of the chaotic attractor of system~\eqref{SysLorenz3}
  with parameters set at $r = 51$, $b = 5.7$, $c =18.3$ by the global absorbing set $\mathcal{B}= \mathcal{B}_R \bigcap \Omega_1 $, where 
  $\mathcal{B}_R$  is the ellipsoid (gray), $\Omega_1$ is the parabolic cylinder (brown).}  
  \label{fig:shap:abs-set}
\end{figure}

The presence of an absorbing set implies the existence of a global attractor $\mathcal{A}_{glob}$, 
which contains all local self-excited 
and hidden attractors 
\cite{Kuznetsov-2020-TiSU,LeonovK-2013-IJBC,DudkowskiJKKLP-2016,LeonovVK-2010-DAN,KuznetsovKLV-2013,DancaFKC-2016-IJBC,KiselevaKL-2016-IFAC,Zhang-C-2019,ZhangF-ZQMXT-2021}
and a stationary set. 
In the interior of the global absorbing set model \eqref{SysShap6} can show both regular and irregular limit dynamics depending upon values of model's parameters \cite{AlexeevaBKM-CSF-2020}. 
In case of the global stability we observe regular dynamics when all trajectories of system \eqref{SysLorenz3} tend to the stationary set $\left\{S_0,\, S_{\pm}\right\}$, where $S_0=(0, 0, 0), \, S_{\pm}=\left(\pm \sqrt{  {b}(  {r}+1)} , \pm \sqrt{  {b} (  {r}+1)},   {r}+1 \right)$ 
  are equilibria of system \eqref{SysLorenz3}. As it was shown in \cite{AlexeevaBKM-CSF-2020}, the system is globally stable in the following parameter domain
\begin{equation}
\left\{
  \begin{aligned}
(  b +1)\left(  \frac{b}{c} -1\right) <   r < \left(  \frac{b}{c} +  1\right)(  b -1) , \\
    2 < b < 2 c . \\
    \end{aligned}
\right. 
\label{CondStabFin_AbsSet2} 
\end{equation}
Thus, in \cite{AlexeevaBKM-CSF-2020} the regular dynamics of system \eqref{SysLorenz3} was studied and the conditions of global stability were obtained.

On the other hand, if condition \eqref{CondStabFin_AbsSet2} is violated, the system exhibits irregular behavior, at which a chaotic attractor can be reveal. As an example, Shapovalov et~al. \cite{ShapovalovKBA-2004,ShapovalKaz-2015}, and Gurina and Dorofeev \cite{GurinaD-2010} show that system \eqref{SysShap6} exhibits chaotic behavior for some values of parameters.

Localization of global attractor and furthest calculation of the limit values of the finite-time Lyapunov exponents and the finite-time Lyapunov dimension along various trjectories of this attractor are nontrivial tasks.
While trivial attractors (stable equilibrium) can be easily found analytically or numerically, the search for periodic and chaotic attractors can be a challenging problem. For numerical localization of the attractor, one needs to choose an initial point in its basin of attraction. After a transient process, a trajectory, starting in a neighborhood of an unstable equilibrium, is attracted to the state of oscillation and then traces it.  Next, the computations are being performed for a grid of points in vicinity of the state of oscillation to explore the basin of attraction and improve the visualization of the attractor.

However, for an arbitrary system possessing a transient chaotic set, the time of transient process depends strongly on the choice of initial data in the phase space and also on the parameters of numerical solvers to integrate a trajectory (e.g., order of the method, step of integration, relative and absolute tolerances). This complicates the task of distinguishing a transient chaotic set from a sustained chaotic set (attractor) in numerical experiments. 
Since the ``lifetime'' of a transient chaotic process can be extremely long and in view of the limitations of reliable integration of chaotic ODEs, even long-time numerical computation of the finite-time Lyapunov exponents and the finite-time Lyapunov dimension does not guarantee a relevant approximation of the Lyapunov exponents and the Lyapunov dimension \cite{KuznetsovLMPS-2018,KuznetsovMKK-2020,AlexeevaBKM-IFAC-2020}.

In this paper, we obtain analytical formula of the exact Lyapunov dimension for global attractor of system \eqref{SysLorenz3}. We demonstrate difficulties in numerical computation of the finite-time Lyapunov exponents and the finite-time Lyapunov dimension along one randomly chosen trajectory over a long time interval which are caused by finite precision numerical integration of ODE, UPOs embedded into the attractor, and choice of various initial data. This confirms the significance of the deduсed analytical formula for the Lyapunov dimension.

\section{Analytical estimation of finite-time Lyapunov dimension and exact Lyapunov dimension}

In this section, we give the main definitions and explanations. Some definitions, proofs and technical parts used from now onwards in this section are summarised in Appendix.

Rewrite system \eqref{SysLorenz3} in a form
\begin{equation} \label{eq:ode}
 \dot u = f(u), \quad f : \mathcal{B} \subseteq \mathbb{R}^3 \to \mathbb{R}^3,
\end{equation}
where $f$ is a continuously differentiable vector-function.
Let $u(t,u_0)$ be any solution of~\eqref{eq:ode} such that $u(0,u_0)=u_0 \in \mathcal{B}$
exists for $t \in [0,\infty)$, it is unique and stays in the absorbing set $\mathcal{B}$. For system \eqref{eq:ode} the evolutionary operator $\varphi^t(u_0) = u(t,u_0)$
defines a smooth \emph{dynamical system}  $\{\varphi^t\}_{t\geq0}$
in the phase space $(U, ||~\cdot~||)$:
$\big(\{\varphi^t\}_{t\geq0},(U \subseteq \mathbb{R}^3,||\cdot||) \big)$,
with Euclidean norm.

We consider fundamental matrix $D\varphi^t(u)=\big(y^1(t),y^2(t),y^3(t)\big)$, $D\varphi^0(u) = I$, with cocycle property, where $\{y^i(t)\}_{i=1}^{3}$ are linearly independent solutions of the linearized system, $I$ is the unit $3\times 3$ matrix.
 The \emph{finite-time local Lyapunov dimension} \cite{Kuznetsov-2016-PLA,KuznetsovLMPS-2018}
can be defined via an analog of the \emph{Kaplan-Yorke formula}
with respect to the set of ordered \emph{finite-time Lyapunov exponents}\footnote{see Appendix.} $\{\LEs_i(D\varphi^t(u))=\LEs_i(t,u)\}_{i=1}^3$ at the point $u$:
\begin{equation}\label{lftKY}
   \dim_{\rm L}(t, u) = d^{\rm KY}(\{\LEs_i(t,u)\}_{i=1}^{3}) =
    j(t,u) + \tfrac{\LEs_1\!(t,u) + \cdots + \LEs_{j(,u)}\!(t,u)
   }{|\LEs_{j(t,u)\!+\!1}(t,u)|},
\end{equation}
where $j(t,u) = \max\{m: \sum_{i=1}^{m}\LEs_i(t,u) \geq 0\}$,
$\dim_{\rm L}(t, u)=3$ for $j(t,u)=3$, or $t=0$.
If $j(t,u) \in \{1,2\}$, then
$\sum_{i=1}^{j(t,u)}\LEs_i(t,u) \geq 0$, $\LEs_{j(t,u)+1}(t,u) <0$
and
\begin{equation}\label{dimsum}
  \dim_{\rm L}(t, u)=j(t,u)+s(t,u) \quad : \quad
  \sum_{i=1}^{j(t,u)}\LEs_i(t,u)+s(t,u)\LEs_{j(t,u)+1}(t,u)=0.
\end{equation}
The \emph{finite-time Lyapunov dimension}
is defined as:
\begin{equation}\label{DOmaptmax}
  \dim_{\rm L}(t, \mathcal{A}) = \sup\limits_{u \in \mathcal{A}} \dim_{\rm L}(t,u),
\end{equation}
where $\mathcal{A}$ is a compact invariant set.

The \emph{Douady--Oesterl\'{e} theorem} \cite{DouadyO-1980} implies that
for any fixed $t > 0$
the finite-time Lyapunov dimension on set $\mathcal{A}$,
defined by \eqref{DOmaptmax},
is an upper estimate of the Hausdorff dimension:
$\dim_{\rm H} \mathcal{A} \leq \dim_{\rm L}(t, \mathcal{A})$.
By the Horn inequality \cite[p.50]{KuznetsovR-2021-book},
cocycle property,
and invariance of $\mathcal{A}$ we have\footnote{see Appendix.} 
$\sup_{u\in \mathcal{A}}\big(\sum_1^{j}\LEs_i(kt,u) + s\LEs_{j+1}(kt,u)\big)\leq 
\sup_{u\in \mathcal{A}}\big(\sum_1^{j}\LEs_i(t,u)+s\LEs_{j+1}(t,u)\big)$
for $j \in \{1,2\}$, $s \in [0,1]$ and any integer $k>0 $.
The infimum is achieved at infinity, otherwise
for $d: 0<\dim_{\rm}(T,\mathcal{A}) < d <\liminf_{k \to +\infty}\dim_{\rm}(kT,\mathcal{A})$
from \eqref{dimsum} and the Horn inequality one gets a contradiction:
\(
0<\liminf\limits_{k \to +\infty}
\sup\limits_{u \in \mathcal{A}}\sum_1^{d}\LEs_i(D\varphi^{kT}(u)) \leq
\liminf\limits_{k \to +\infty}\sup\limits_{u \in \mathcal{A}}\sum_1^{d}\LEs_i(D\varphi^{T}(u))<0
\).
Thus, the best estimation \eqref{DOmaptmax} takes
the form \cite{Kuznetsov-2016-PLA}
\begin{equation}\label{defLD}
   \dim_{\rm L} \mathcal{A}
   = \inf_{t > 0}\sup\limits_{u \in \mathcal{A}} \dim_{\rm L}(t, u)
   = \liminf_{t \to +\infty}\sup\limits_{u \in \mathcal{A}} \dim_{\rm L}(t, u)
\end{equation}
and is called the \emph{Lyapunov dimension}. 

If the supremum of finite-time local Lyapunov dimensions on set $\mathcal{A}$
is achieved at such an equilibrium point
$u_{eq} \equiv \varphi^{t}(u_{eq}) \in \mathcal{A}$:
$\dim_{\rm L}\mathcal{A}=\dim_{\rm L}u_{eq}$,
then the Lyapunov dimension can be represented in analytical form
and it is called  \emph{exact Lyapunov dimension} in \cite{DoeringGHN-1987}. A conjecture on the Lyapunov dimension of self-excited attractor 
\cite{Kuznetsov-2016-PLA,KuznetsovLMPS-2018,KuznetsovMKK-2020} is that for a typical system, the Lyapunov dimension of a self-excited attractor does not exceed the Lyapunov dimension of one of the unstable equilibria, the unstable manifold of which intersects with the basin of attraction and visualizes the attractor.

In a general case, analytical computation of the Lyapunov exponents and Lyapunov dimension is hardly possible.
However, they can be estimated by the eigenvalues of the symmetrized Jacobian matrix \cite{DouadyO-1980,Smith-1986}. The Kaplan–Yorke formula with respect to the ordered set of eigenvalues $\nu_i(J(u))=\nu_i(u)$, $\nu_1(u) \ge \nu_2(u) \ge \nu_3(u)$, $i=1,2,3,$ of the symmetrized Jacobian matrix $\frac{1}{2} (J(u)+J(u)^*)$,  $J(u)=Df(u)$ \cite{Kuznetsov-2016-PLA} gives an upper estimation of the Lyapunov dimension of an attractor~$\mathcal{A}$: 
\begin{equation}\label{dimestviaJ}
 \dim_{\rm L}\mathcal{A}=
 \inf_{t > 0}\sup\limits_{u \in \mathcal{A}}
 d^{\rm KY}(\{ \LEs_i(t,u)\}_{i=1}^{3})
 \leq
 \sup_{u \in \mathcal{A}}d^{\rm KY}\big(\{\nu_i(u)\}_{i=1}^3\big).\\
\end{equation}
Generally speaking, one cannot get the same values of $\{\nu_i(u)\}_{i=1}^{3}$ at different points $u$; thus, the supremum of $d^{\rm KY}(\{\nu_i(u)\}_{i=1}^{3})$ on $\mathcal{A}$ has to be computed. To obtain estimate \eqref{dimestviaJ}, it is not necessary to integrate the solutions of the system; however, the analytical estimation of $\{\nu_i(u)\}_{i=1}^3$
on the attractor may be a challenging task.
Another approach is based on \emph{the Leonov method} of analytical estimation of the Lyapunov dimension\footnote{see Appendix.}.
The inequality $
   \dim_{\rm H}\mathcal{A} \leq
    \dim_{\rm L}\mathcal{A}
   < j+s$ holds, 
if 
 \begin{equation}\label{ineq:weilSVct}
  \sup_{u \in \mathcal{A}} \big(
  \nu_1 (u,S) + \cdots + \nu_j (u,S)
  + s\nu_{j+1}(u,S) + \dot{V}(u)
  \big) < 0,
\end{equation}
where $\dot{V} (u) = ({\rm grad}(V))^{*}f(u)$, $V: U \subseteq \mathbb{R}^3 \to \mathbb{R}^1$ is a differentiable scalar function, $S$ is a nonsingular $3\times 3$ matrix, $\nu_i(u,S)=\nu_i(SJ(u)S^{-1})$ is the ordered set of eigenvalues $\nu_1(u,S) \ge \nu_2(u,S) \ge \nu_3(u,S)$, $i=1,2,3,$ of the symmetrized Jacobian matrix $\frac{1}{2} (SJ(u)S^{-1}+(SJ(u)S^{-1})^*)$, $j \in \{1, 2\}$ is an integer number, and $s \in [0,1]$ is a real number. 

\section{Main result}

Using an effective analytical approach, proposed by Leonov \cite{Leonov-1991-Vest,Kuznetsov-2016-PLA}, which is based on a combination of the Douady-Oesterl\'{e} approach with the direct Lyapunov method we estimate the Lyapunov exponents and obtain the Lyapunov dimension for the global attractor in system \eqref{SysLorenz3}.

\begin{theorem}\label{theorem:shapovalov:stability}
If for parameters of system \eqref{SysLorenz3} the following relations hold

\begin{equation}
     2 < b < 2 c,  \\
   \label{Cond:AbsSet}
\end{equation}
\begin{equation}
     r > \left(  \frac{b}{c} +  1\right)(  b -1) , 
   \label{Cond:Dim-r}
\end{equation}
\begin{equation}
  (b+1)\bigl[(b-2)(b^2+6bc-3c^2+b)+c(2c-b)\bigr] - c\Bigl(b^2+b-c(8-b)\Bigr)r \le 0,
   \label{Cond:Dim-reg}
\end{equation}
then 

\begin{equation}
\dim_{\rm L}\mathcal{A}_{glob}=3-\frac{2(b+c-1)}{c-1+\sqrt{(c+1)^2+4cr}}.
   \label{DimL}
\end{equation}

\end{theorem}
{}
\begin{proof}
Consider system \eqref{SysLorenz3} with the Jacobian matrix 

\begin{equation}
J=\begin{pmatrix} -  c &  c &0\\
 r- z & 1 &- x\\
y & x& -  b\end{pmatrix}
\label{JacLor} 
\end{equation}
under the conditions \eqref{Cond:AbsSet} and \eqref{Cond:Dim-r}. We apply the transformation \eqref{TrfShLor} with a nonsingular matrix

\begin{equation}
 S=\begin{pmatrix}\frac{-1}{a} &0 &0\\
-\frac{  b  + 1}{  c} &1 &0\\
0& 0& 1\end{pmatrix}
\label{MatrixS}
\end{equation}
to this system, where $a=\frac{c}{\sqrt{\left(1+  b\right)\left(  c-  b\right)+ r   c}}$.
Then the symmetrized Jacobian matrix of this system $\frac{1}{2}\left(SJS^{-1}+(SJS^{-1})^*\right)$\footnote{
Symbol $^*$ denotes the transposition of matrix.
} has the following eigenvalues

\begin{equation}
\lambda_2=-  b, \, \lambda_{1,3}=-\frac{  c -1}{2}\pm \frac{1}{2}\left( (2  b +1 -   c)^2 + a^{2}\left(\frac{  b +1}{  c}x+y\right)^2 + \left(a z-\frac{2  b}{a}\right)^2\right)^\frac{1}{2}.
\label{EigValMS}
\end{equation}
The inequalities

\begin{equation}
2(\lambda_j - \lambda_{j+1}) \ge -(-1)^j (2 b +1- c)  + | 2b + 1 - c|\ge 0,\quad j={1,2},
\end{equation}
imply $\lambda_1\ge\lambda_2\ge\lambda_3$. From \eqref{EigValMS} following \cite{Leonov-1991-Vest} we get the ratio

\begin{equation}
2(\lambda_1 + \lambda_2  + s\lambda_3)=-(s+1)(c -1) - 2b  + (1-s) \left( (2  b +1 -   c)^2 + a^{2}\left(\frac{  b +1}{c}x+y\right)^2 + \left(a z-\frac{2  c}{a}\right)^2\right)^\frac{1}{2},
\label{SumLambda2}
\end{equation}
where $s \in [0,1]$ is a real number.
Using the famous inequality $\sqrt{k+l} \le \sqrt{k} + \frac{l}{2\sqrt{k}}$, $\forall k>0, \, l\ge 0$, we obtain an estimate

\begin{equation}
\begin{aligned}
&2(\lambda_1 + \lambda_2  + s\lambda_3)\le -(  c -1 + 2  b) - s(  c - 1) + (1-s)\left[ (  c+1)^2+4  c r\right]^{\frac{1}{2}}+\\
&+\frac{2(1-s)}{\left[ (  c+1)^2+4  c r\right]^\frac{1}{2}}
\left[-  c z+\frac{a^2 z^2}{4} +\frac{a^2}{4}
\left(\frac{b + 1}{c}x+y\right)^2\right]. \end{aligned}
\label{ESumLambda2}
\end{equation}

We introduce the function $V(x,y,z)=\frac{\theta(x,y,z)}{\left[ (c+1)^2+4  c r\right]^{\frac{1}{2}}}$, where

\begin{equation}
\theta(x, y, z) = a^2 Q_0 x^2 + a^2(- c Q_1 + Q_2) y^2 + a^2Q_2 z^2 + \frac{a^2}{4c} Q_1 x^4 - a^2 Q_1 x^2 z - a^2 P Q_1 xy -  \frac{c}{b}z,
\end{equation}
$P$ and $Q_i \, (i=\overline{0,2})$ are some positive real parameters.
Then

\begin{equation}
\begin{aligned}
&2(\lambda_1 + \lambda_2  + s\lambda_3)+2\dot V \le -(  c -1 + 2  b) - s(  c - 1) + (1-s)\left[ (  c+1)^2+4  c r\right]^{\frac{1}{2}}+\\
&+\frac{2(1-s)}{\left[ (  c+1)^2+4  c r\right]^\frac{1}{2}}\left[W(x,y,z)+\dot \theta\right],\end{aligned}
\label{ESumLambda2V}
\end{equation}
where $W(x,y,z)=-  c z+\frac{a^2 z^2}{4} +\frac{a^2}{4}\left(\frac{ b + 1}{c}x+y\right)^2$. Choose the parameters $P$ and $Q_i \, (i=\overline{0,2})$ of the function $\theta(x, y, z)$  such that 
\begin{equation}
F:=W(x,y,z)+\dot \theta \le 0, \qquad \forall x,y,z\ge \frac{x^2}{2c}. \,\\
\label{Wtheta}
\end{equation}
Substituting $W(x,y,z)$ and $\dot{\theta}$ in \eqref{Wtheta}, we get

\begin{equation}
F=A_0 z^2 + A_1 x^2 + A_2 xy + A_3 y^2,  
\label{F:Qform}
\end{equation}
where 
\begin{equation}
  \begin{aligned}
 & A_0 = a^2\Biggl(2c\left(b+P\right)Q_1-2bQ_2 + \frac{1}{4}\Biggr),\\
 & A_1 = a^2\Biggl(\frac{(b+1)^2}{4c^2} -r P Q_1\Biggr),\\
 & A_2 = a^2\Bigl[\Bigl((c-1)P-2c\Bigr)Q_1+2rQ_2 + \frac{b+1}{2c} - \frac{c}{ba^2}\Bigr],\\
 & A_3 = a^2\Bigl(\frac{1}{4}+2Q_2-c(2+P)Q_1\Bigr).
   \end{aligned}
\label{QformPar}
\end{equation}
Then
\begin{equation}
\left.
  \begin{aligned}
 & A_0 \le 0\\
 & A_3 \le 0\\
 & 4A_1A_3 - A_2^2 \ge 0  
  \end{aligned}
\right\}\Rightarrow F\le 0,\quad \forall x,y,z\ge \frac{x^2}{2c},
\label{Main-condition}
\end{equation}

\begin{equation}
\Leftrightarrow\left\{
  \begin{aligned}
 & Q_1 \le \frac{b}{c(b+P)} Q_2 - \frac{1}{8c(b+P)}, \\
 & Q_1 \ge \frac{2}{c(2+P)} Q_2 + \frac{1}{4c(2+P)}, \\
 & Q_1 \ge \frac{2}{2c+P} Q_2 + \frac{(b+c+1)^2 b a^2 - 4 c^3}{4a^2bc^2(r+1)(2c+P)}.
  \end{aligned}
\right.
\label{Mainsys}
\end{equation}
Since RHS of the second inequality in \eqref{Mainsys} is positive, we obtain

\begin{equation}
\left\{
  \begin{aligned}
 &\frac{b}{c(b+P)} Q_2 - \frac{1}{8c(b+P)} - \left(\frac{2}{c(2+P)} Q_2 + \frac{1}{4c(2+P)}\right)\ge 0, \\
 & \frac{b}{c(b+P)} Q_2 - \frac{1}{8c(b+P)} -\left(\frac{2}{2c+P} Q_2 + \frac{(b+c+1)^2 b a^2 - 4 c^3}{4a^2bc^2(r+1)(2c+P)}\right)\ge 0,
  \end{aligned}
\right.
\label{Def-R-L1-L2-1}
\end{equation}

\begin{equation}
\Leftrightarrow\left\{
  \begin{aligned}
 &\frac{(b-2)P}{c(b+P)(2+P)} Q_2 - \frac{3P+2b+2}{8c(b+P)(2+P)}\ge 0, \\
 & -\frac{(2c-b)P}{c(b+P)(2c+P)} Q_2 + \\
 &\frac{\Bigl(c(8c-b) r - 2b^3-4(3c+1)b^2-(2+13c-6c^2)b+8c^2\Bigr)P+}{8bc^2(b+P)(2c+P)(r+1)}\\
 &\frac{+6bc^2 r -2b(b+1)(b^2+b+6bc-3c^2)}{}\ge 0. 
  \end{aligned}
\right.
\label{Def-R-L1-L2-2}
\end{equation}
It follows from condition \eqref{Cond:AbsSet} that the coefficient at $Q_2$ in the first inequality of \eqref{Def-R-L1-L2-2} is positive and the coefficient at $Q_2$ in the second inequality of \eqref{Def-R-L1-L2-2} is negative. Hence, we can reduce \eqref{Def-R-L1-L2-2} to the following inequalities
\begin{equation}
L(b,c,r,P)\le Q_2\le R(b,c,r,P),
\label{Inequality-RL}
\end{equation}
where $L(b,c,r,P)=\frac{3P+2b+2}{8P(b-2)}>0$,\\
 $R(b,c,r,P)=
\frac{\Bigl(c(8c-b) r - 2b^3-4(3c+1)b^2-(2+13c-6c^2)b+8c^2\Bigr)P+6bc^2 r -2b(b+1)(b^2+b+6bc-3c^2)}{8bc(2c-b)(r+1)P}$.\\
Inequalities \eqref{Inequality-RL} mean that a positive $Q_2$ exists such
that
\begin{equation}
R(b,c,r,P)-L(b,c,r,P)=-\frac{(b+P)(k_1 r+k_0)}{4bc(b-2)(2c-b)(r+1)P}\ge0,
\label{Inequality:R-L}
\end{equation}
where $k_1=-c\Bigl(b^2+b-c(8-b)\Bigr)$, $k_0=(b+1)\bigl[(b-2)(b^2+6bc-3c^2+b)+c(2c-b)\bigr]$.
Since the denominator of fraction \eqref{Inequality:R-L} is positive, we obtain required condition \eqref{Cond:Dim-reg}

\begin{equation}
k_1 r+k_0\le 0.
\label{MainInequality}\\
\end{equation}
This completes the proof.
\end{proof}

We obtain a formula of the exact Lyapunov dimension of the global attractor for certain region of the parameters $(b,\, c)$ (Fig.~\ref{fig:reg-b-c}) of system \eqref{SysLorenz3}. The same approach allows one to estimate of the topological entropy of the global attractor \cite{AdlerKA-1965,PogromskyM-2011,KuznetsovKKMD-2019-IFAC,KuznetsovR-2021-book}.

\begin{figure}[ht]
  \centering
  \includegraphics[width=0.40\linewidth]{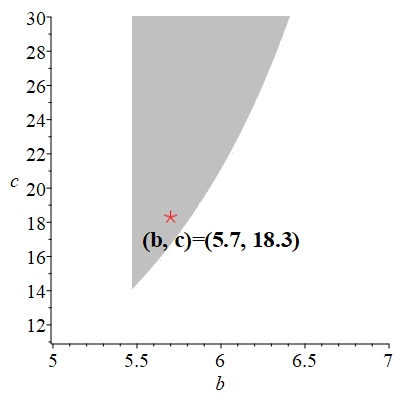}
  \caption{Parameters of system~\eqref{SysLorenz3} complying with the conditions \eqref{Cond:AbsSet} and \eqref{Cond:Dim-reg}
  .}
  \label{fig:reg-b-c}
\end{figure}

To demonstrate significance of this analytical result we compare it with numerical simulations. We discuss the difficulties of numerical procedures for reliable estimation of the Lyapunov dimension and Lyapunov exponents along one randomly chosen trajectory over a long time interval.
A natural way to get reliable estimation of the Lyapunov dimension of attractor $\mathcal{A}$ is to localize the attractor $\mathcal{A}\subset \mathcal{C}$, to consider a grid of points $\mathcal{C}_{grid}$ on $\mathcal{C}$, and to find the maximum of the corresponding finite-time local Lyapunov dimensions for a certain time $t = T$. 
We show the grid of points $\mathcal{C}_{\rm grid}$
filling the basin of attraction:
the grid of points fills cuboid
$\mathcal{C} = [-27,27] \times [-65,65] \times [3,95]$ (containing the attractor)
rotated by 45 degrees around the $z$-axis,
with the distance between points equal to $0.5$ (see Fig.~\ref{fig:shapovalov:grid}).
The time interval considered is $[0,\, T=500]$ at the time points $t = t_k = \tau\,k$ $(k=1,\ldots,N)$, $N= 1000$
according to the time step $\tau = t_k - t_{k-1} = 0.5$, and the integration method is MATLAB ode45 with predefined parameters.
The infimum on the time interval
is computed at the points $\{t_k\}_{1}^{N}$.

For system \eqref{SysLorenz3} with parameters under consideration,
we use a MATLAB realization of \emph{the adaptive algorithm of the finite-time Lyapunov dimension
and Lyapunov exponents computation}~\cite{KuznetsovLMPS-2018}
and obtain the maximum of the finite-time local Lyapunov dimensions at the grid of points
\big($\displaystyle\max_{u \in \mathcal{C}_{\rm grid}} \dim_{\rm L}(t,u)$ is
computed for trajectories of system \eqref{SysLorenz3}
using MATLAB ode45 integration method with predefined parameters 
and with threshold parameter $\delta = 0.01$
for adaptively adjusting the number of SVD approximations\big).
For parameters $r = 51$, $b = 5.7$, $c = 18.3$
we get

\begin{equation}\label{num-dim}
\max_{u \in \mathcal{C}_{\rm grid}} \dim_{\rm L}(100,u)=2.0808, \\ 
 \max_{u \in \mathcal{C}_{\rm grid}} \dim_{\rm L}(500,u)=2.0792.
\end{equation}
Note that if conditions on dissipativiness of system \eqref{SysLorenz3} are not satisfied for a certain time, $t=t_k$ the computed trajectory is out of the cuboid,
the corresponding value of the finite-time local Lyapunov dimension
is not taken into account in the computation of the maximum
of the finite-time local Lyapunov dimension
(e.g. if there are trajectories with initial conditions in the cuboid,
which tend to infinity).

\begin{figure}[ht]
  \centering
  \includegraphics[width=0.5\linewidth]{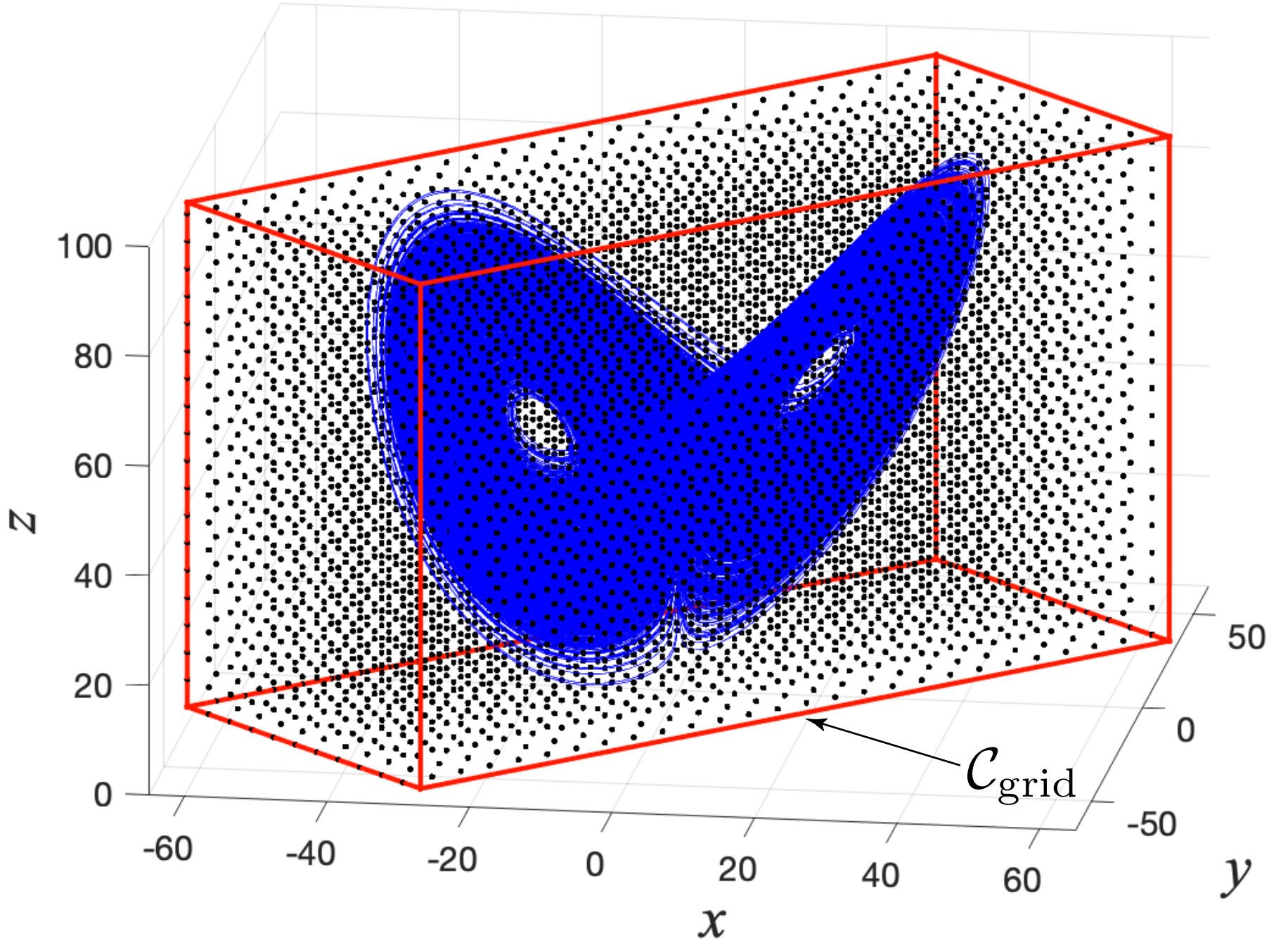}
  \caption{Numerical localization of the chaotic attractor of system~\eqref{SysLorenz3}
  with parameters set at $r = 51$, $b = 5.7$, $c = 18.3$ by the
  cuboid $\mathcal{C}$ and the
  corresponding grid of points $\mathcal{C}_{\rm grid}$.}
  \label{fig:shapovalov:grid}
\end{figure}

If the maximum of local Lyapunov dimensions on the global attractor, which involves all equilibria, is achieved at an equilibrium point: $\dim_{\rm L}(u_{eq}^{cr}) = \max_{u\in \mathcal{A}} \dim_{\rm L}(u)$, then this allows one to get analytical formula for the exact Lyapunov dimension \cite{DoeringGHN-1987}.	

The exact Lyapunov dimension $\dim_{\rm L}\!\mathcal{A}_{{}glob}=\dim_{\rm L}\!S_0=2.4347> \dim_{\rm L}\!\mathcal{A}\approx \max_{u \in \mathcal{C}_{\rm grid}}\!\dim_{\rm L}(t_k,u)\\
\approx 2.0808$ (see \eqref{num-dim}) obtained by formula \eqref{DimL} and the estimation \eqref{num-dim} are consistent with the hypothesis on the Lyapunov dimension of self-excited attractors. Using Theorem~\ref{theorem:shapovalov:stability} we can get the value of the exact Lyapunov dimension on the global attractor, which coincides with the Lyapunov dimension at a stationary (zero) point. This result is nontrivial since to compute reliably numerically the dimensions on the trajectories of the global attractor is extremely difficult. We demonstrate challenging nature of this task by the following examples.

Choosing the initial data somewhere in the phase space, we can obtain the values of the dimensions along the various trajectories  by a numerical procedure. Generally speaking, these values of the dimensions will also be different. 
For instance, system \eqref{SysLorenz3} has the analytical solution $u(t) = (0, 0, z_0 e^{-bt})$ which tends to the equilibrium $S_0 = (0, 0, 0)$ from any initial point $(0, 0, z_0) \in \mathbb{R}^{3}$. The existence of such solutions in the phase space complicates the procedure of visualization of a chaotic attractor (pseudo-attractor) by one pseudo-trajectory with arbitrary initial data computed for a sufficiently large time interval. 
In particular, the numerical computation of finite-time local Lyapunov exponents along this trajectory during any time interval does not lead to averaging of these values across the attractor, but to tending of these values to the finite-time local Lyapunov exponents of $S_0$.

The challenges of the finite-time Lyapunov dimension computation along the trajectories over large time intervals is connected with the existence of UPOs embedded into a chaotic attractor. The ``skeleton'' of a chaotic attractor for this system comprises embedded UPOs. Along with the existence of the analytical solution $u(t) = (0, 0, z_0 e^{-bt})$ the global attractor of system \eqref{SysLorenz3} contains a period-1 UPO.

Consider system \eqref{eq:ode}.
Let $u^{\rm upo}(t,u^{\rm upo_1}_0)$ be its UPO with period $\tau > 0$,
$u^{\rm upo}(t - \tau,u^{\rm upo_1}_0) = u^{\rm upo}(t,u^{\rm upo_1}_0)$,
and initial condition $u^{\rm upo_1}_0=u^{\rm upo}(0,u^{\rm upo_1}_0)$.
To compute the UPO, we add the unstable delayed feedback control (UDFC) \cite{Pyragas-1992} in the following form:
\begin{equation}\label{eq:closed_loop_syst}
 \begin{aligned}
\dot{u}(t) &= f(u(t)) - K B \, \big[F_N(t) + w(t)\big], \\
\dot{w}(t) &= \lambda_c^0 w(t) + (\lambda_c^0 - \lambda_c^\infty) F_N(t), \\
F_N(t) &= C^*u(t) - (1\!-\!R) \sum_{k=1}^N R^{k-1} C^*u(t - k T),
  \end{aligned}
\end{equation}
where
$0 \leq R < 1$ is an extended DFC parameter,
$N = 1,2,\ldots,\infty$ defines the number of previous states involved in
delayed feedback function $F_N(t)$, 
$\lambda_c^0 > 0$, and $\lambda_c^\infty < 0$ are
additional unstable degree of freedom parameters,
$B, C$ are vectors and $K > 0$ is a feedback gain.
For the initial condition $u^{\rm upo_1}_0$ and $T = \tau$ we have
$F_N(t) \equiv 0, \ w(t) \equiv 0$,
and, thus, the solution of system \eqref{eq:closed_loop_syst}
coincides with the periodic solution
of the initial system \eqref{eq:ode}.

For system~\eqref{SysLorenz3} with parameters $r = 51$, $b = 5.7$, $c = 18.3$, 
using~\eqref{eq:closed_loop_syst} with
$B^{*} = \left(0, 1, 0\right)$, $C^* = \left(0, 1, 0\right)$,
$R = 0.7$, $N =100$, $K = 10$, $\lambda^0_c = 0.1$, $\lambda^\infty_c = -5$,
one can stabilize a period-1 UPO $u^{\rm upo_1}(t,u_0)$
with period $\tau_1 = 0.69804$
from the initial point $u_0 = (0.1, 0.1, 0.1)$, $w_0 = 0$ on the time interval $[0, 100]$ (see Fig.~\ref{fig:shapovalov:upo1attr}). 
\begin{figure}[ht]
\begin{center}
     \includegraphics[width=0.5\linewidth]{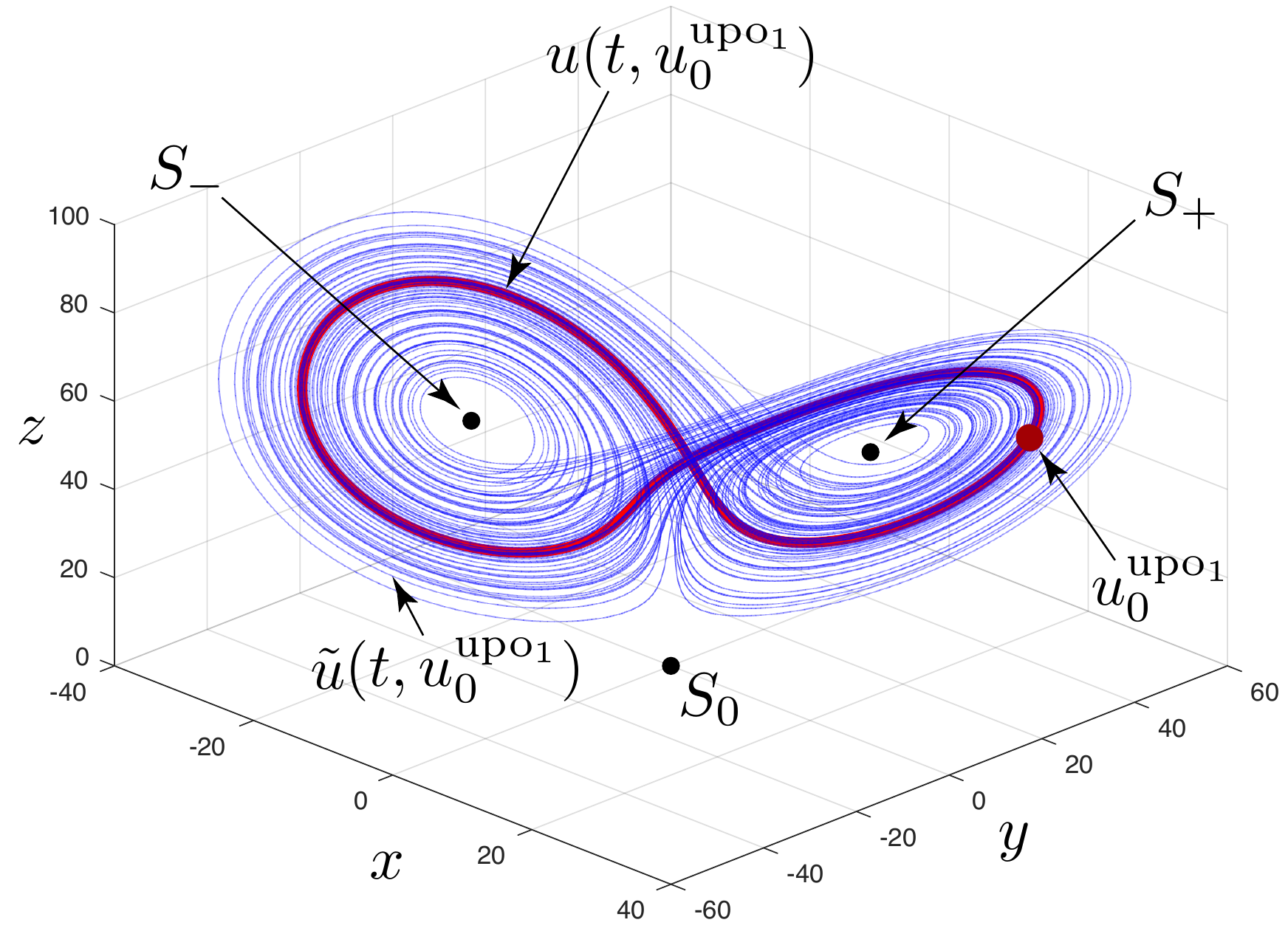}
    \caption{Period-1 UPO $u^{\rm upo_1}(t)$ (red, period $\tau_1 = 0.69804$)
    stabilized using the UDFC method,
    and pseudo-trajectory $\tilde{u}(t, u^{\rm upo_1}_0)$ (blue, $t \in [0,100]$) in system \eqref{SysLorenz3} with parameters set at $r = 51$, $b = 5.7
    $, $c = 8.3$. 
}
    \label{fig:shapovalov:upo1attr}
 \end{center}
 \end{figure}
We use the Pyragas procedure \cite{Pyragas-1992,Pyragas-2001} for numerical stabilization and visualization of UPOs.
For the initial point $u^{\rm upo_1}_0 \approx (29.6688, 26.1650, 73.8221)$
on the UPO $u^{\rm upo_{1}}(t) = u(t, u^{\rm upo_1}_0)$
we numerically compute the trajectory
of system \eqref{eq:closed_loop_syst} without the stabilization (i.e. with $K = 0$)
on the time interval $[0,T=100]$ (see Fig.~\ref{fig:shapovalov:upo1attr}).
We denote it by $\tilde{u}(t, u^{\rm upo_1}_0)$
to distinguish this pseudo-trajectory from the periodic orbit $u(t, u^{\rm upo_1}_0)$.
On the initial small time interval $[0,T_1 \approx 2 \tau_1]$,
even without the control,
the obtained trajectory $\tilde{u}(t, u^{\rm upo_1}_0)$
approximately traces the ''true'' trajectory (periodic orbit) $u(t, u^{\rm upo_1}_0)$.
But for $t > T_1$, without a control,
the pseudo-trajectory $\tilde{u}(t, u^{\rm upo_1}_0)$
diverges from $u(t, u^{\rm upo_1}_0)$ and
visualize a local chaotic attractor~$\mathcal{A}$.

In general, the closeness of the real trajectory $u(t,u_0)$
and the corresponding pseudo-trajectory $\tilde u(t,u_0)$
calculated numerically can be guaranteed on a limited short time interval only. The obtained values of the largest \emph{finite-time Lyapunov exponent} ${\rm LE}_1(t,u_{0}^{upo_1})$ computed along the stabilized UPO $u(t, u^{\rm upo_1}_0)$ and the trajectory without stabilization $\tilde{u}(t, u^{\rm upo_1}_0)$ gives us the following  results. On the initial part of the time interval $[0,T_1 \approx 2 \tau_1]$, one can indicate the coincidence of these values with a sufficiently high accuracy. 
After $t > T_2\approx 10$ the difference in values becomes significant and the corresponding graphs diverge in such a way that the graph corresponding to the unstabilized trajectory is higher than the parts of the graphs corresponding to the UPO and the analytical value largest Lyapunov exponent: ${\rm LE}_1(u_{0}^{upo_1})=1.80401$, computed via Floquet multipliers (see Fig.~\ref{fig:shapovalov:LE}). 

\begin{figure}[ht]
\begin{center}
     \includegraphics[width=0.5\linewidth]{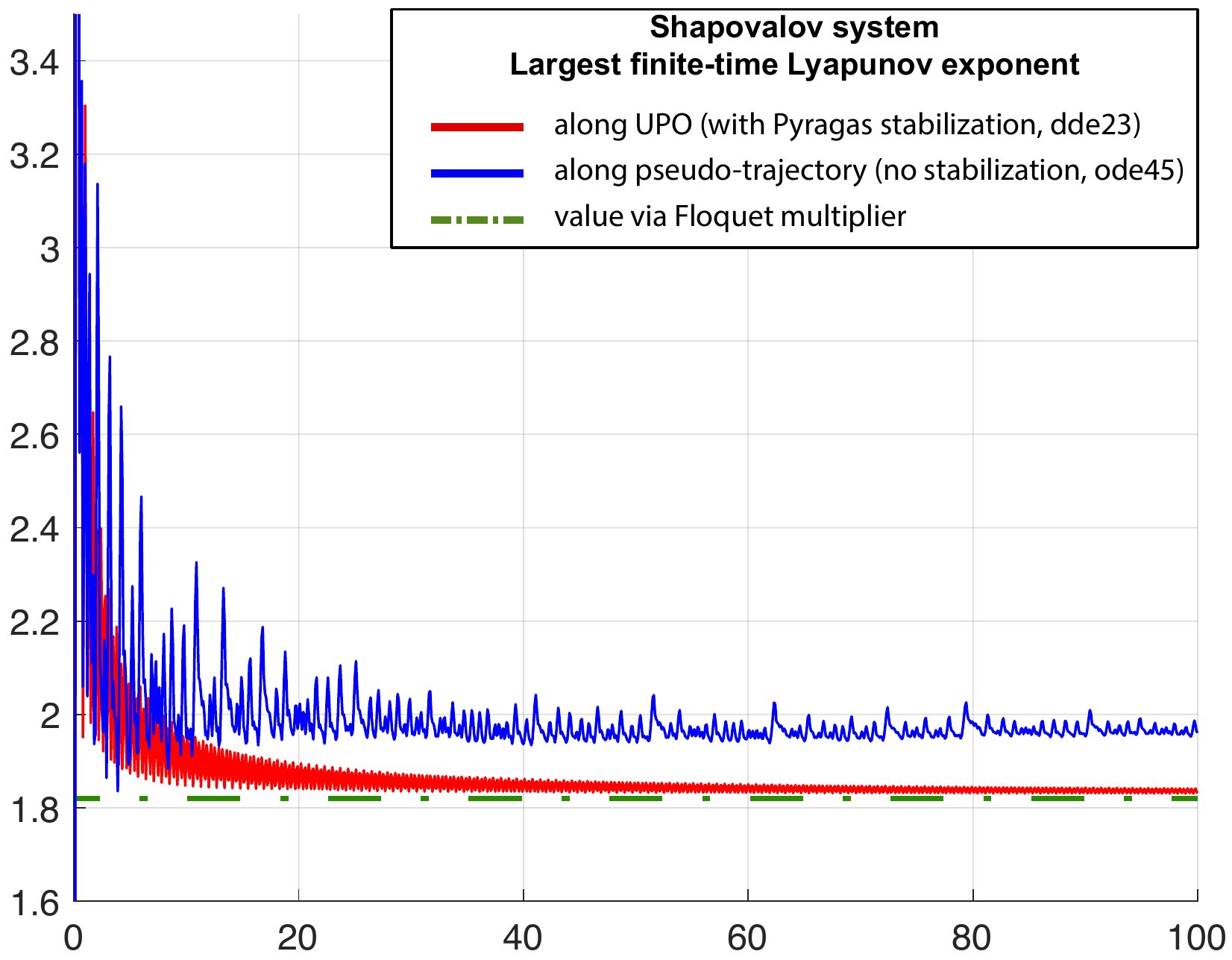}
    \caption{Period-1 UPO $u^{\rm upo_1}(t)$ (red, period $\tau_1 = 0.69804$)
    stabilized using the UDFC method,
    pseudo-trajectory $\tilde{u}(t, u^{\rm upo_1}_0)$ (blue), and the analytical value ${\rm LE}_1(u_{0}^{upo_1})$ (green) for $t \in [0,100]$ in system \eqref{SysLorenz3} with parameters set at $r = 51$, $b = 5.7$, $c = 18.3$.  }
    \label{fig:shapovalov:LE}
 \end{center}
 \end{figure}

Using numerical experiments, we analyze the chaotic dynamics of system \eqref{SysLorenz3} and visualize a self-excited attractor for values of parameters $r=51, b=5.7, c=18.3$. At the same time we get formula of the exact Lyapunov dimension of the global attractor for certain region of the parameters $(b,\, c,\, r)$  \eqref{Cond:AbsSet} and \eqref{Cond:Dim-r} of system \eqref{SysLorenz3}. Thus, we get the following relations 
\begin{equation}\label{ineq-dimL}
\dim_{\rm L}\mathcal{A}_{glob}=\dim_{\rm L}S_0=2.4347>\!\dim_{\rm L}\mathcal{A}\!\approx\!\max_{u \in \mathcal{C}_{\rm grid}}\!\dim_{\rm L}(t_k,u)\!\approx 2.0808\!\ge\!\dim_{\rm L}u^{upo_1}\!\approx 2.0738.
\end{equation}

\section*{Conclusion}

In this paper, we studied the irregular behavior (chaotic attractor, unstable limit cycles) of the mid-size firm model, assuming the deterministic endogenous mechanism for generating these fluctuations. Using an analytical approach, we calculated quantitative characteristics of irregular dynamics, such as the Lyapunov dimension and topological entropy, and demonstrated the complexity and ambiguity of using numerical procedures for calculating these indicators. We have obtained a number of new results. First, we proved a theorem about the exact formula for the Lyapunov dimension of the global attractor in the model. Similar way used for getting the formula of the topological entropy. Second, we identified an UPO for the model and stabilized it using the Pyragas control procedure. Third, we numerically calculated the finite-time Lyapunov dimension along the trajectories of the global attractor, including UPO, thereby providing support for arguments about difficulties of application of the numerical procedures and importance of the obtained exact formula for the Lyapunov dimension of the global attractor. 
We believe that expanding our knowledge of the role, sources, as well as qualitative and quantitative characteristics of irregular oscillatory dynamics may diminish researchers' reliance on unrealistically large shocks to explain economic data.

\section*{Acknowledgments}

This paper was prepared with the support by the Leading Scientific Schools of Russia: project NSh-2624.2020.1 (sections~3,4). 
Authors from the St.Petersburg State University
acknowledge support from St.Petersburg State University grant Pure ID 75207094 (section~1-2).

This work was motivated by research conducted at the Institute for Nonlinear Dynamical Inference at the International Center for Emerging Markets Research (http://icemr.ru/institute-for-nonlinear-dynamical-inference/) and Financial Research Institute of the Ministry of Finance of the Russian Federation,
a number of whose employees the authors thank for for helpful suggestions and comments.
Especially, we thank William A. Barnett, with whom the authors began work on the paper \cite{AlexeevaBKM-CSF-2020}, for his extremely valuable comments and support.





\newpage

\section*{Appendix}

Let $\{\varphi^t\}_{t\geq0}$ denote a smooth \emph{dynamical system} with continuous time, and let set $\mathcal{A}$ be its compact invariant set.
Fundamental matrix $D\varphi^t(u)=\big(y^1(t),y^2(t),y^3(t)\big)$, $D\varphi^0(u) = I$ consists of linearly independent solutions $\{y^i(t)\}_{i=1}^{3}$ of the linearized system, where $I$ is the unit $3\times 3$ matrix with the following cocycle property:
\begin{equation}\label{cocycle}
  D\varphi^{t+s}(u) = D\varphi^t\big(\varphi^s(u)\big)D\varphi^{s}(u),
  \ \forall t,s \geq 0, \ \forall u \in U.
\end{equation}

Let $\LEs_i(\cdot) = t^{-1} \ln\sigma_i(\cdot)$ for $t>0$, where $\sigma_i(D\varphi^t(u))=\sigma_i(t,u)$, $i=1,2,3$, be the singular values of $D\varphi^t(u)$
(i.e. $\sigma_i(t,u)>0$ and $\sigma_i(t,u)^2$ are the eigenvalues of the symmetric matrix $D\varphi^t(u)^*D\varphi^t(u)$ with respect to their algebraic multiplicity), ordered so that
$\sigma_1(t,u)\geq \sigma_2(t,u) \geq \sigma_3(t,u) > 0$ for any $u \in U$, $t\geq0$.
Consider a set of \emph{finite-time Lyapunov exponents}
$\{\LEs_i(D\varphi^t(u))=\LEs_i(t,u)\}_{i=1}^3$ at the point $u$:
\begin{equation}\label{ftLE}
  \LEs_i(t,u) = \frac{1}{t}\ln\sigma_i(t,u), \ t > 0, \quad i=1,2,3,
\end{equation}
ordered by decreasing for all $t>0$. 
We can introduce the following concepts -- \emph{finite-time local Lyapunov dimension}
(of map $\varphi^t$ at point $u$):
$\dim_{\rm L}(t,u) = \dim_{\rm L}(\varphi^t, u)$,
the \emph{finite-time Lyapunov dimension} (of map $\varphi^t$
with respect to set $\mathcal{A}$):
$\dim_{\rm L}(t,\mathcal{A}) = \dim_{\rm L}(\varphi^t, \mathcal{A})$,
and for the \emph{Lyapunov dimension}
(of dynamical system $\{\varphi^t\}_{t\geq0}$
with respect to set $\mathcal{A}$):
$\dim_{\rm L}\mathcal{A} = \dim_{\rm L}(\{\varphi^t\}_{t \geq0}, \mathcal{A})$.

Consider the dynamical system
$\big(\{\varphi^t\}_{t\geq0},(U\subseteq \mathbb{R}^3,||\cdot||) \big)$
under the change of coordinates $w = h(u)$,
where $h: U \subseteq \mathbb{R}^3 \to \mathbb{R}^3$ is a diffeomorphism.
In this case the dynamical system
$\big(\{\varphi^t\}_{t\geq0},(U\subseteq \mathbb{R}^3,||\cdot||) \big)$
is transformed to the dynamical system
$\big(\{\varphi_h^t\}_{t\geq0},(h(U)\subseteq \mathbb{R}^3,||\cdot||) \big)$,
and the compact set $\mathcal{A} \subset U$
invariant with respect to $\{\varphi^t\}_{t\geq0}$
is mapped to the compact set $h(\mathcal{A}) \subset h(U)$.
Here
\begin{equation}\label{Dphih}
  D\varphi_h^t(w)=
  Dh(\varphi^t(u))
  D\varphi^t(u)
  \big(Dh(u)\big)^{-1}.
\end{equation}

\begin{proposition}\label{thm:dDOunderdiff}\label{thm:hdiffLE} (see, e.g. \cite{Kuznetsov-2016-PLA,KuznetsovAL-2016})
For any diffeomorphism $h: U \subseteq \mathbb{R}^3 \to \mathbb{R}^3$
the Lyapunov dimension is invariant
with respect to diffeomorphism, i.e.
\begin{equation}\label{dDOunderdiff}
  \dim_{\rm L}(\{\varphi^t\}_{t\geq0},\mathcal{A})
  =
  \dim_{\rm L}(\{\varphi_h^t\}_{t\geq0},h(\mathcal{A})).
\end{equation}
\end{proposition}

The proof of this proposition uses the Horn inequality for \eqref{Dphih}
and the fact that
singular values of $Dh(\varphi^t(u))$ and $(Dh(\varphi^t(u)))^{-1}$
are uniformly bounded in $t$ on $\mathcal{A}$.
Moreover, instead of $Dh$ one can consider any $\,3\times 3 $ matrix $H(u)$, such that all its elements are scalar continuous functions of $u$ and $\det H(u) \neq 0$ for all $u \in \mathcal{A}$, and get\footnote{By the Horn inequality for the matrices $D_H(\varphi^t(u))=
  H(\varphi^t(u))D\varphi^t(u) H(u)^{-1}$ and $D\varphi^t(u)
  =H(\varphi^t(u))^{-1}  D_H(\varphi^t(u))H(u).$}
 
\begin{equation}\label{dDOunderH}
\begin{aligned}
  \lim\limits_{t \to +\infty}
  \bigg( \LEs_i\big( H(\varphi^t(u))D\varphi^t(u)\big(H(u)\big)^{-1} \big) -
  \LEs_i\big(D\varphi^t(u)\big)\bigg)=0,
  \quad \quad i=1,2,3, \\
  \dim_{\rm L}(\{\varphi^t\}_{t\geq0},\mathcal{A})
  =
  \liminf_{t \to +\infty}\sup\limits_{u \in \mathcal{A}}d^{\rm KY}\big(
  \{\LEs_i\big( H(\varphi^t(u))D\varphi^t(u)\big(H(u)\big)^{-1}\}_1^3 \big) \big).
\end{aligned}
\end{equation}

If an equilibrium $u_{eq} \! \equiv \! \varphi(u_{eq}) \! \in \! \mathcal{A}$ has
simple real eigenvalues, then a nonsingular $3~\times~3$ matrix $S$ exists
such that the linearisation takes the form
\(  SDf(u_{eq})S^{-1}
    ={\rm diag}\big(\lambda_1(u_{eq}),\cdots\!,\lambda_3(u_{eq})\big)
\),
where  $\lambda_{j}(u_{eq}) \geq \lambda_{j+1}(u_{eq})$, $i=1,2$. 
Then, by the linear change of variables $w=h(u)=Su$ 
and the invariance we get
$\lim\limits_{t\to +\infty}\LEs_i(t,u_{eq}) = \lambda_{i}(u_{eq})$
and
$\dim_{\rm L}u_{eq} = d^{\rm KY}(\{ \lambda_{i}(u_{eq})) \}_{i=1}^{3}$.

For analytical estimation of the Lyapunov dimension via the eigenvalues of the symmetrized Jacobian matrix we use the generalized Liouville's relation
(see, e.g., \cite{Smith-1986},\cite[p.68]{KuznetsovR-2021-book})
and get, $\forall t>0, \ u \in \mathcal{A}$, the following:
\begin{equation}\label{LEviaJ}
 \sum\limits_{i=1}^{j}\LEs_i(\varphi^t(u))\!+\!s\LEs_{j+1}(\varphi^t(u))
 \!\leq\! 
 \frac{1}{t}\!\int\limits_0^t\!\sum\limits_{i=1}^{j}\nu_i(\varphi^\tau(u))
 \!+\!s\nu_{j+1}(\varphi^\tau(u)) d\tau
 \!\leq\!
 \sup\limits_{u \in \mathcal{A}} \sum\limits_{i=1}^{j}\nu_i(u)\!+\!s\nu_{j+1}(u).\\
\end{equation}
From \eqref{LEviaJ} we obtain the upper estimation of the Lyapunov dimension \eqref{dimestviaJ}.

\emph{The Leonov method} of analytical estimation of the Lyapunov dimension is based on \eqref{dDOunderH} and \eqref{dimestviaJ}. Following \cite{Leonov-1991-Vest,LeonovL-1993,Leonov-2002}, we consider
$H(u)=p(u)S$, where $p: U \subseteq \mathbb{R}^3 \to \mathbb{R}^1$
is a continuous scalar function, $S$ is a nonsingular $3\times 3$ matrix. Then we compute the Lyapunov dimension by~\eqref{dDOunderH}:
\[
  \dim_{\rm L}\mathcal{A}
    =
  \liminf_{t \to +\infty}\sup\limits_{u \in \mathcal{A}}d^{\rm KY}\big(
  \{ \LEs_i\big(p(\varphi^t(u))p(u)^{-1}\,SD\varphi^t(u)S^{-1}\big)\}_1^3 \big),
\]
and estimate it by \eqref{dimestviaJ}.
For that
by \eqref{ftLE} and \eqref{LEviaJ} we get the estimation:
\begin{equation}\label{partsumLE}
  \sum\limits_{i=1}^{j}
  \LEs_i\big(p(\varphi^t(u))p(u)^{-1}SD\varphi^t(u)S^{-1}\big)
 \leq j\frac{1}{t}\ln \big(p(\varphi^t(u))p(u)^{-1}\big)
  +\frac{1}{t}\int_0^t\sum\limits_{i=1}^{j}\nu_i(SJ(u)S^{-1}) d\tau.
\end{equation}
In general, while
under the diffeomorphism $h(u)=Su$
the Lyapunov dimension is invariant
and $J(u) \to SJ(u)S^{-1}$, the values
$\nu_i(SJ(u)S^{-1})=\nu_i(u,S)$
are not invariant and, thus, $S$ together with $p(u)$
may be used to simplify their computation
(the idea with $S$ was introduced in
\cite[eq.(8)]{LeonovL-1993} and $p(u)$ was introduced in \cite{Leonov-1991-Vest}).
The scalar multiplier of the type $p(\varphi^t(u))(p(u))^{-1}$ can be interpreted as the changes of Riemannian metrics \cite{NoackR-1996} (see, also \cite{KuznetsovR-2021-book}).
The following theorem is a reformulation
of the results from \cite{Leonov-2002,Leonov-2012-PMM}
(see also \cite{Kuznetsov-2016-PLA,KuznetsovR-2021-book}).

\begin{theorem}\label{thm:LD-estimate-V}
If there exist an integer $j \in \{1,2\}$,
a real $s \in [0,1]$, a differentiable scalar function $V: U \subseteq \mathbb{R}^3 \to \mathbb{R}^1$, and a nonsingular $3\times 3$ matrix $S$
such that
\begin{equation}\label{ineq:weilSVct}
  \sup_{u \in \mathcal{A}} \big(
  \nu_1 (u,S) + \cdots + \nu_j (u,S)
  + s\nu_{j+1}(u,S) + \dot{V}(u)
  \big) < 0,
\end{equation}
where $\dot{V} (u) = ({\rm grad}(V))^{*}f(u)$,
then
$
   \dim_{\rm H}\mathcal{A} \leq
    \dim_{\rm L}\mathcal{A}
   < j+s.
$
\end{theorem}
\begin{proof}
Let $p(u) = e^{V(u)(j+s)^{-1}}$. Then
\(
  (j+s)\frac{1}{t}\ln\big(p(\varphi^t(u))p(u)^{-1}\big)
  =
  \frac{1}{t}\left( \int_{0}^{t} \dot V(\varphi^\tau(u)) d\tau \right).
\)
Thus by invariance of $\mathcal{A}$ and \eqref{LEviaJ} from \eqref{partsumLE}
we get
\begin{equation}\label{weiladd}
\begin{aligned}
 \sum\limits_{i=1}^{j}\LEs_i(SD\varphi^t(u)S^{-1})+ s\LEs_{j+1}(SD\varphi^t(u)S^{-1} )
 +(j+s)\frac{1}{t}\ln\big( p(\varphi^t(u))p(u)^{-1} \big)\leq \\
 \leq \sup_{u \in \mathcal{A}} \bigg( \sum\limits_{i=1}^{j}\nu_i(u,S)
 +s\nu_{j+1}(u,S) + \dot V(u) \bigg)<0.
\end{aligned}
\end{equation}
 Since  
 \(
   \lim\limits_{t \to +\infty}(j+s)\frac{1}{t}\ln\big(p(\varphi^t(u))p(u)^{-1}\big)
   =0
\) for any $u \in \mathcal{A}$
there exists 
  $T>0$ such that
\begin{equation}\label{weil}
 \sum\limits_{i=1}^{j}\LEs_i(SD\varphi^t(u)S^{-1}) +
 s\LEs_{j+1}(SD\varphi^t(u)S^{-1}) <0 \quad \forall t>T, \ u \in \mathcal{A}.
\end{equation}
Thus, taking into account \eqref{dimsum}, $\dim_{\rm L}\mathcal{A}<j+s$.
\end{proof}

\end{document}